\documentclass[12pt]{article}

\usepackage{epsfig}
\usepackage{amssymb}
\usepackage{amsfonts}

\usepackage{color}

\def\be{\begin{equation}}
\def\ee{\end{equation}}
\def\ba{\begin{array}{c}}
\def\ea{\end{array}}

\def\ben{$$}
\def\een{$$}

\newcommand{\bea}{\begin{eqnarray}}
\newcommand{\eea}{\end{eqnarray}}

\newcommand{\bbr}{\br\!\br}
\newcommand{\kkt}{\kt\!\kt}

\newcommand{\kt}{\rangle}
\newcommand{\br}{\langle}

\newtheorem{thm}{Theorem}

\newtheorem{lemma}[thm]{Lemma}

\newenvironment{proof}{\noindent {\bf Proof}}{\hfill$\square$\vspace{3mm}\endtrivlist}

\begin{document}

\titlepage

 \begin{center}{\Large \bf

Construction of maximally non-Hermitian potentials
under unbroken ${\cal PT}-$symmetry constraint

  }\end{center}

 \begin{center}

\vspace{8mm}

  {\bf Miloslav Znojil} $^{1,2,3}$

\end{center}

\vspace{8mm}

  $^{1}$
 {The Czech Academy of Sciences,
 Nuclear Physics Institute,
 Hlavn\'{\i} 130,
250 68 \v{R}e\v{z}, Czech Republic, {e-mail: znojil@ujf.cas.cz}}


 $^{2}$
 {Department of Physics, Faculty of
Science, University of Hradec Kr\'{a}lov\'{e}, Rokitansk\'{e}ho 62,
50003 Hradec Kr\'{a}lov\'{e},
 Czech Republic}

  $^{3}$
Institute of System Science, Durban University of Technology,
Durban, South Africa


\subsection*{Keywords}.

.

systems with spontaneously unbroken PT-symmetry;

non-Hermitian exceptional-points degeneracy extreme;

access to quantum phase transition;

discrete imaginary potentials;


\subsection*{Abstract}

A family of discrete
Schr\"{o}dinger equations
with imaginary and maximally non-Hermitian
multiparametric
potentials $V(x,A,B,\ldots)$
is studied.
In the
domain of unitarity-compatible parameters ${\cal D}$
(known as the
domain of spontaneously unbroken ${\cal PT}-$symmetry)
the reality
of all of the bound-state energies $E_n(A,B,\ldots)$
survives up to
the maximally non-Hermitian
``exceptional-point'' (EP)
extreme with
parameters $A^{(EP)},B^{(EP)},\ldots$.
The computer-assisted
proof of existence and the
symbolic-manipulation
localization of such a
spectral-degeneracy limit
are sampled showing that their
complexity grows quickly with
the number of para\-meters.

\newpage

\section{Introduction}

The concept of ${\cal PT}-$symmetry of
a quantum Hamiltonian
$H$
(i.e., relation $H{\cal PT}={\cal PT}H$
with suitable ${\cal P}$ and ${\cal T}$)
became extremely popular
after the Bender's and Boettcher's
well known and
influential
analysis \cite{BB} of several
ordinary differential ${\cal PT}-$symmetric
Schr\"{o}dinger equations
 \be
 -\psi''(x) +V(x)\psi(x)=E\psi(x)\,
 \label{BBSE}
 \ee
in which
the components of symmetry
${\cal P}$ and ${\cal T}$
represented the
parity and time reversal, respectively.
On these \textcolor{black}{grounds} these authors proposed that
even when the
Hamiltonian
happens to be manifestly non-Hermitian,
the spectrum of energies
can remain real.
The
${\cal PT}-$symmetry
can be then declared (spontaneously)
unbroken \cite{Carl}
and one can speak about various innovative ``analytic continuations
of conventional theories'' \cite{BB}.

The Bender's and Boettcher's
choice of their complex
but still safely confining potentials
has been initially inspired by
the properties of the
imaginary cubic oscillator (ICO) with
 \be
 V(x)=V^{(ICO)}(x)={\rm i}x^3\,
 \label{lico}
 \ee
and by
the Bessis' never published
conjecture of reality of the related
bound-state-like energies.
Incidentally, the
Bessis' intuitive
reluctance of publishing
his conjecture
\cite{DB}
found an {\it a posteriori\,} validation
in 2012 when
Siegl with Krej\v{c}i\v{r}\'{\i}k
\cite{Siegl}
formulated their statement and proof
that
``there is no quantum mechanical Hamiltonian''
associated with the
``${\cal PT}-$symmetric imaginary cubic oscillator''
represented by Eqs.~(\ref{BBSE}) + (\ref{lico}).
In 2019, moreover,
G\"{u}nther with Stefani \cite{Uwe} added that
``what is still lacking is a simple physical explanation
scheme for the non-Rieszian behavior of the [ICO]
eigenfunction sets''.

An explanation can be found provided by
older literature
(cf. the Dieudonn\'{e}'s
critical note \cite{Dieudonne})
and, in particular, by one of the
physics-oriented but still sufficiently rigorous
review \cite{Geyer}
in which the authors
recommended to consider just a
bounded-operator
subclass
of the general non-Hermitian Hamiltonians
with real spectra.
Obviously, IOP does not satisfy such a condition
(in this respect, interested readers might also check
our most recent related comments in \cite{foundations}).

In \cite{Geyer}, Scholtz et al
established the
``normal quantum-mechanical interpretation''
of the latter \textcolor{black}{bounded} Hamiltonians.
We
decided to follow the recommendation.
In what follows we will replace, for this reason,
the ordinary differential Eq.~(\ref{BBSE})
by its difference-equation simplification
(cf. the first half of section \ref{themo} for details).
Due to this simplification, we will be able to
initiate the study of
at least some of the
above-mentioned ICO-connected pathologies
from a different and, perhaps, slightly
broader methodical perspective.
The main goal
of our approach aimed at
an insight in
some aspects of
these pathologies
will be formulated in the second half
of section \ref{themo}.

In section \ref{obser}
we will
recall
a few comprehensive reviews of the field \cite{ali,book,SIGMA}.
We will summarize some of the basic theoretical ingredients
forming the
necessary update of the formulation of unitary
quantum mechanics
using the fundamental concept of unbroken ${\cal PT}-$symmetry.
We will illustrate this update by its application to
the special case of
our toy-model
Hamiltonian of secton \ref{themo}
defined as acting in
$N-$dimensional Hilbert space with $N=2$.

In the literature, paradoxically, one finds practically no
analogous results
with the choice of $N=3$.
An explanation is provided in
section \ref{locika} where we
turn attention to the model
of section \ref{themo} with $N=3$.
We will show there that
some of the features of
such a next-to-trivial model
(which still varies with just a single parameter)
are more subtle than expected.

What follows is a systematic study
of our model with  $N=4$ and $N=5$
(in section \ref{relocika}) and with  $N\geq 6$
(in section \ref{rererelocika} covering some specific features of
the larger-matrix models).
A brief discussion and conclusions are then added in section
\ref{presumma}.


\section{Bound states in purely imaginary potentials\label{themo}}

The methodical guidance
as provided by multiple reviews \cite{Carl,ali,book}
opens the way towards
a consistent return of many unitary
(i.e., so called ``closed'')
quantum systems and models using
non-Hermitian but
${\cal PT}-$symmetric
Hamiltonians $H$
back to the
safe and tested framework of
conventional textbooks.
An explicit illustration of such a statement
will be offered
in section \ref{obser}. A
compact outline of the abstract theory
will be facilitated there
by several restrictions of
our present specific choice
to the model.

\subsection{Simplification:
Discrete version of Schr\"{o}dinger equation}

First of all, we decided to
replace, for our present purposes,
the ordinary differential ICO
Schr\"{o}dinger equation (or,
in general,
any ${\cal PT}-$symmetric
Eq.~(\ref{BBSE})
with any local potential)
by its difference-equation analogue
 \be
 -
 \psi_n(x_{k-1}) +2\,\psi_n(x_k)
 -
 \psi_n(x_{k+1})
 +V(x_k)\,\psi_n(x_k)=E_n\,\psi_n(x_k)\,,
 \ \ \ \ \ n=1, 2,\ldots\,, N\,.
 \label{divSE}
 \ee
The wave function
is assumed to
live here on a discrete,
finite and equidistant grid-point lattice
of coordinates
 \be
 x_k = x_0+k\,\delta\,,\ \ \ \ k=0,1,\ldots, N+1
 \label{lakoon}
 \ee
with $\delta=1$.
The physical asymptotics of wave functions \textcolor{black}{are}
specified by the most common
Dirichlet
boundary conditions,
 \be
 \psi_n(x_0)= \psi_n(x_{N+1})=
 0\,.
 \ee
This means that
in the search for solutions
we will just have to diagonalize
the $N$ by $N$ matrix Hamiltonians
which are defined as sums
 \be
 H^{(N)}(A,B,\ldots)=\triangle^{(N)}+V^{(N)}(A,B,\ldots)
 \label{hamacek}
 \ee
of
a (shifted) one-dimensional discrete Laplacean
 \be
 \triangle^{(N)} =
 \left[ \begin {array}{ccccc}
  0&-1&0
 &\ldots&0
 \\
 {}-1&0&-1&\ddots&\vdots
 \\
 {}0&-1&\ddots&\ddots
 &0
 \\
 {}\vdots&\ddots&\ddots&0&-1
 \\
 {}0&\ldots&0&-1&0
 \end {array} \right]\,.
 \label{a0t}
 \ee
with any suitable ${\cal PT}-$symmetric
interaction potential.

For our present purposes we will feel guided by the
continuous
ICO problem
so that we will restrict our attention, for the sake of
definiteness, just to the models characterized by the
purely imaginary local interactions,
 \be
 V^{(N)}(A,B,\ldots)=
 \left[ \begin {array}{ccccc}
  -i\,A&0&0
 &\ldots&0
 \\
 {}0&-i\,B&0&\ddots&\vdots
 \\
 {}0&0&\ddots&\ddots
 &0
 \\
 {}\vdots&\ddots&\ddots&i\,B&0
 \\
 {}0&\ldots&0&0&i\,A
 \end {array} \right]\,.
 \label{aVto}
 \ee
The resulting tridiagonal-matrix Hamiltonian will
vary, in general, with as many as $
[N/2]$ variable real parameters.
In what follows, their variability
will serve, first of all, as a tool
of a controlled increase of non-Hermiticity
of the Hamiltonian.

The unusual multiparametric flexibility
of our model (\ref{hamacek}) -- (\ref{aVto})
will enable us to address
and discuss several currently
open questions, with our main
attention being paid
to the possible, ${\cal PT}-$symmetry mediated
compatibility between a maximalization
of the non-Hermiticity
under the constraint of
survival of the
unitarity of the evolution or, in technical terms, of the
reality and non-degeneracy of the spectrum.

\subsection{Extremes of non-Hermiticity}

Contradictory as the latter compatibility may seem to be,
it may be characterized as one of the key innovations
provided by the use of the notion of unbroken
${\cal PT}-$symmetry.
It
opens the possibility of
various fairly counterintuitive forms of
correspondence between
a parameter-controlled increase of
non-Hermiticity of the Hamiltonian
(whatever such an increase may mean)
and the robustness (or survival)
of the reality of the spectrum,
whatever such a manifestation
of the unbroken ${\cal PT}-$symmetry
might mean
in operational sense.

Our main ambition will be
to treat the reality of the spectrum
as a fixed constraint.
Under this constraint
we will try to specify
a subdomain of
parameters $A, B, \ldots$
in which the non-Hermiticity would be,
in some reasonably well defined sense,
maximal.
We will see that even
at the first few nontrivial $N=2,3,\ldots$
the variability of interaction (\ref{aVto})
will
open the way towards the answer.
Its essence
will be found to lie in
the emergence of the Kato's
exceptional point (EP, \cite{Kato})
of a maximal order $N$ (i.e., EPN)
at the related degenerate bound-state energy $E=E^{(EPN)}=0$.

The latter type of a very specific
non-Hermitian degeneracy
can be perceived as
mimicking
the above mentioned
\textcolor{black}{ICO-related}
non-Rieszian behavior
of the
basis in Hilbert space
\textcolor{black}{with $N=\infty$. A}
successful
simulation of
such an ``intrinsic exceptional point'' (IEP, \cite{Siegl})
at a finite $N$
will require a localization
of the  parameters
$A=A^{(EPN)},B=B^{(EPN)},\ldots$ at which
the Hamiltonian matrix
 \ben
 H^{(EPN)}=\triangle^{(N)}+V^{(EPN)}(A^{(EPN)},B^{(EPN)},\ldots)
  \een
becomes
non-diagonalizable,
canonically represented
by the
$N$ by $N$ Jordan matrix
 $$
 J^{(N)}=\left[ \begin {array}{ccccc}
                     0&1&0&\ldots&0
 \\{}0&0&1&\ddots&\vdots
 \\{}\vdots&0&\ddots&\ddots&0
 \\{}\vdots&\vdots&\ddots&0&1
 \\{}0&0&\ldots&0&0
 \end {array}
 \right]\,
 $$
in such a way that
 \be
 J^{(N)}=[{Q}^{(N)}]^{-1}\,\cdot H^{(EPN)}\,\cdot {Q}^{(N)}\,
\label{relaJ}
 \ee
where the isospectrality/similarity map ${Q}^{(N)}$ is
called transition matrix.

\section{Unitary
theory\label{obser}}

The EPN-related
non-diagonalizability of $H^{(EPN)}$
(i.e., the existence of the transition matrices in (\ref{relaJ}))
can only be achieved
at the parameters  $A^{(EPN)},B^{(EPN)},\ldots$
which can significantly deviate
from the values of the above-mentioned
ICO potential $V(x)=ix^3$
at the respective grid-point coordinates $x_j$.
In this sense the possibility of a direct connection
between the two models remains an open question
at present.
Nevertheless,
according to the abstract theory
which we are now going to outline,
the first steps towards an answer may be seen
in
the possibility of
making our discrete models hiddenly Hermitian.,

\subsection{The notion of hidden Hermiticity\label{suobser}}

The variability of parameters $A, B, \ldots$
enables us to \textcolor{black}{achieve,}
in principle at least,
the reality of all of
the eigenvalues $E=E(A, B, \ldots)$ of our
toy model Hamiltonian
$H=H^{(N)}(A,B,\ldots)$.
We may assume that
there exists a non-empty domain of parameters ${\cal D}$
in which the
${\cal PT}-$symmetry remains unbroken,
i.e., in which
the spectrum
remains real and non-degenerate.
Then, the key task
\textcolor{black}{for} a theoretician
is to show that
one can still speak about
a quantum system admitting the conventional
probabilistic interpretation.

One of the oldest affirmative answers has been
provided by Scholtz et al \cite{Geyer}.
In the resulting formalism
one is allowed to admit that
in a consistent definition of a
closed (i.e., unitary) quantum system
the observable quantities can be
introduced
in
a fairly unconventional non-Hermitian operator representation.

In the resulting (often called ``quasi-Hermitian'')
equivalent
reformulation of the entirely standard
quantum mechanics
(cf., e.g., its
reviews \cite{Carl,ali,SIGMA})
even the most common requirement of unitarity
acquires a perceivably counterintuitive form.
In this formalism, indeed,
the observables (with real spectra)
become represented by
certain manifestly non-Hermitian operators
(say, $\Lambda_j \neq \Lambda_j^\dagger$, $j=0,1,\ldots$)
which are only required to be quasi-Hermitian \cite{Dieudonne},
 \be
 \Lambda_j^\dagger\,\Theta=\Theta\,\Lambda_j\,,
 \ \ \ \ j=0,1,\ldots,
 \label{quasih}
 \ee
i.e., Hermitian with respect
to an
{\it ad hoc\,} inner-product metric
$\Theta = \Theta^\dagger >0$
which remains $j-$independent \cite{Geyer}.
In other words, ``the normal quantum-mechanical interpretation''
\cite{Geyer} of the model
is obtained, provided only that we replace the conventional inner
product $\br \psi_a|\psi_b\kt$ in the conventional preselected
Hilbert space ${\cal H}$ by its upgrade
 \be
 \br \psi_a|\Theta|\psi_b\kt
 \label{forumro}
 \ee
\textcolor{black}{(to be abbreviated as
$\bbr \psi_a|\psi_b\kt$ in what folllows)}
where the operator $\Theta$ is called ``metric''.

The apparent paradox of coexistence
of the overall theoretical unitarity-of-the-evolution requirement
with a manifest non-self-adjointness of the observable can be easily
clarified when one realizes that the self-adjointness of an operator
can very easily be lost (and also re-acquired) after a change of the
inner product in the Hilbert space of states.
For the
Hamiltonians having real spectra and acting in a suitable Hilbert
space,
the property of their non-self-adjointness can simply
be treated as inessential, attributed to the mere ``ill-chosen''
(i.e., mathematically preferred but manifestly unphysical) inner
product.

The use of the nontrivial physical
Hilbert-space metric $\Theta$
enables us to treat the upgrade of inner product
 $\br \psi_a|\psi_b\kt\ \to \ \br \psi_a|\Theta|\psi_b\kt
 $
as a replacement of the preselected and user-friendly but manifestly
unphysical Hilbert space
 $
 {\cal H}={\cal H}_{mathematical}
 $
by its amended, formally non-equivalent and $\Theta-$dependent
alternative
 $
 {\cal H}={\cal H}_{physical}={\cal H}_{\Theta}
 $.
In such a formulation,
one of the decisive motivations of the introduction
of a nontrivial $\Theta$
has to be seen
in its adaptability to
the variable dynamical parameters characterizing, say,
the Hamiltonian.

This implies, first of all, that the spectral reality
of the observables
need not be robust.
The flexibility of metric
can be used to admit an access to
a quantum phase transition
{\it alias\,}  ``quantum catastrophe'' \cite{catast}.
At the critical EPN value of
parameters the quantum system in question changes its initial
character and loses its observability status.

\subsection{The case of hiddenly Hermitian
Hamiltonian with $N=2$\label{rwobser}}

Probably the best known elementary example of a non-Hermitian but
${\cal PT}-$symmetric quantum Hamiltonian
with real spectrum
is the two-level model
  \be
 H^{(2)}(A)=
 \left[ \begin {array}{cc} -iA&-1\\{}-1&iA\end {array} \right]\,,
 \ \ \ \ \ A \in \mathbb{R}
 \label{ham2a}
 \ee
for which the pair of eigenvalues
 $$
 E^{(2)}_\pm(A)=\pm \sqrt{1-A^2}
 $$
remains real if and only if $-1\leq A \leq 1$.

The two ends of the interval can be recognized as representing the
pair of the Kato's exceptional points EP=EP2. At these
points, in spite of the reality of the eigenvalues (or, more
precisely, of the roots of the related secular equation
$E^2+A^2-1=0$), the matrix itself ceases to be diagonalizable and,
hence, tractable as an observable quantum Hamiltonian.
Thus, whenever
the value of $A$ leaves the physical spectral-reality
open interval ${\cal D}^{}=(-1,1)$
the ${\cal PT}-$symmetry becomes spontaneously broken and
one can speak about a quantum phase transition.

Needless to add that
once we select, say, the positive value of $A^{(EP2)}_{+}= 1$,
we may recall formula (\ref{relaJ}) and
obtain
the canonical representation of the EP2-degenerate
model with Jordan block
 $$
 J^{(2)}=\left[ \begin {array}{cc} 0&1\\{}0&0\end {array} \right]
 $$
and with the transition matrix
 \be
 Q^{(2)}=\left[ \begin {array}{cc} -i&1\\{}-1&0\end {array}
 \right]\,.
 \label{tra2}
 \ee
What is less obvious is the validity of the
condition of quasi-Hermiticity (\ref{quasih})
in application to the
Hamiltonian \textcolor{black}{at $A < A_+^{(EP2)}$,}
 \be
 H^\dagger\,\Theta=\Theta\,H\,.
 \label{joher}
 \ee
As long as $N=2$ it is easy to show that such a
quasi-Hermiticity {\it alias\,} ``Hermiticity constraint in
disguise'' becomes
\textcolor{black}{satisfied if and only if we choose}
the inner product metrics
\textcolor{black}{in the one-parametric form}
 \be
 \Theta=\Theta^{(2)}({A},\xi)=
 \left[ \begin {array}{cc} 1&\xi-iA\\{}\xi+iA&1\end {array}
 \right]\,
 \label{metr2}
 \ee
with a free parameter $\xi$ and
with
the $\xi-$dependent
eigenvalues $\theta_\pm = 1\pm \sqrt{A^2+\xi^2}$.

\textcolor{black}{As} long as the metric must remain positive definite
we have to restrict the range of $\xi$
to an interaction-dependent interval of
 $$
 \xi \in (-\sqrt{1-A^2},\sqrt{1-A^2})\,.
 $$
We can conclude that
under this condition
the quantum system in question remains unitary
and observable even
in the two maximally non-Hermitian dynamical regimes  \textcolor{black}{lying}
in an arbitrarily small vicinity of one of the two
eligible EP2 extremes, i.e., \textcolor{black}{in the regimes where}
$ A \lessapprox +1$ or $ A \gtrapprox -1$,
with
the singularities at $A=A^{(EP2)}_\pm =\pm 1$ excluded.

\section{Non-uniqueness of
interpretation ($N=3$)\label{locika}}

One of the rarely mentioned consequences of
the $(N-1)-$parametric ambiguity of the inner product metric
as sampled by formula (\ref{metr2})
(for arbitrary $N$ see also \cite{SIGMAdva})
is that
$\Theta=\Theta_{physical}(H)\neq I$
can be kept nontrivial
even when the
corresponding Hamiltonian
is Hermitian,
$H_{unusual} = H_{unusual}^\dagger$.
For illustration of such an anomaly
the
toy-model Hamiltonian (\ref{ham2a}) with $N=2$
is too elementary
so that
we are forced to select
its next, $N=3$ descendant
 \be
 H^{(3)}=\triangle^{(3)}+V^{(3)}(A)=\left[ \begin {array}{ccc} -iA&-1&0
 \\{}-1&0&-1\\{}0&-1&iA\end {array} \right]\,.
 \label{han3}
 \ee
In the literature, incidentally, the latter model
is much less frequently discussed, in spite of having
the comparably elementary formulae for
its triplet of energy levels
including trivial $E_0=0$ and the doublet
$E_\pm=\pm \sqrt{2-A^2}$.

One can still
expect the existence of non-Hermitian degeneracies of course.

\begin{lemma}
\label{lemmaprvni}
Hamiltonian (\ref{han3}) possesses
two Kato's EP3 singularities at $A^{(EP3)}_\pm = \pm \sqrt{2}$.
\end{lemma}
\begin{proof}.
At $N=3$ the constructive
proof of the existence of the Kato's EP3
form of the degeneracy
is still straightforward,
based
on Eq.~(\ref{relaJ})
and on the brute-force linear-algebraic
evaluation
(i.e., the constructive proof of existence)
of the respective
transition matrices.
Thus, in a way paralleling the $N=2$ result (\ref{tra2})
we obtain
 $$
 Q^{(3)}=\left[ \begin {array}{ccc} -1&-i\sqrt {2}&1
\\{}i\sqrt {2}&-1&0\\{}1&0&0\end {array} \right]\,
 $$
at the positive critical coupling $A^{(EP3)}_{+}= \sqrt{2}$, etc.
\end{proof}

\subsection{Hermitization of weakly non-Hermitian potentials}

During our search
for the physical (i.e., Hamiltonian-dependent)
inner product metric at $N=3$
we had to keep in mind
that
in a way sampled by the structure of
formula (\ref{metr2})
the admissible metrics $\Theta^{(3)}$
must form,
up to an inessential overall premultiplication
factor, a two-parametric family
\cite{SIGMAdva}.
Still,
in comparison with the preceding
$N=2$ construction,
we were surprised
by the emergence of a few qualitatively new obstructions.

Initially we
decided to
follow the methodical guidance as provided
by the $N=2$ model.
We
assigned the role of parameters
to the real parts of matrix elements
$\Theta^{(3)}_{12}$ and $\Theta^{(3)}_{13}$ and
we obtained the following result.

\begin{lemma}
\label{druhelemma}
At $N=3$, relation (\ref{joher})
is satisfied by the two-parametric family of
the 3 by 3 matrix metrics
 \be
 \Theta^{(3)}(A,\xi,\eta)=
  \left[ \begin {array}{ccc}
   1&{\eta}-iA&{\xi}-iA{\eta}
  \\\noalign{\medskip}{\eta}+iA&{\xi}+1+{ A}^{2}&{\eta}-iA
 \\\noalign{\medskip}{\xi}+iA{\eta}&{\eta}+iA&1
 \end {array} \right]\,.
 \label{ondula}
 \ee
\end{lemma}
\begin{proof}
is based on the brute-force solution
of Eq.~(\ref{joher}) which is
interpreted as a set of nine algebraic equations
determining the unknown matrix elements of the metric.
\end{proof}

Next, we applied
the requirement of a minimization of anisotropy
of the metric
as recommended in \cite{Lotor}.
This led us
to the choice of $\xi=\eta=0$ in (\ref{ondula}).
We obtained
the following triplet of eigenvalues,
 \be
 \theta_0=1\ \ \ {\rm and}
\ \ \ \theta_\pm=1 + \frac{1}{2}
\left [ A^2 \pm \sqrt{8A^2+A^4}
\label{xylo}
\right ]\,.
 \ee
The (necessary)
positive definiteness of the metric
is only guaranteed
for $0 \leq A < 1$.
Thus,
the applicability
of the result of Lemma \ref{druhelemma}
with $\xi=\eta=0$ is found
restricted
to the models $H^{(3)}(A)$
in which the non-Hermiticity is
far from maximal, i.e., far from the dynamical
regime of our present interest.

\subsection{Hermitization of strongly non-Hermitian potentials}

With our attention turned to the latter dynamical regime we may
assume that
 $$
 1 \leq A < \sqrt{2}\,.
 $$
The use of formula~(\ref{ondula}) with trivial
$\xi=0$ and $\eta=0$
becomes inapplicable
in such a vicinity of the
EP3 singularity.
We have to employ a different method of construction
of the metric, especially
in the truly
strongly non-Hermitian dynamical regime where
$ A \lessapprox \sqrt{2}$.

We will use the method which
has been proposed and tested
in
\cite{PLA}.
Its basic idea lies in the factorization ansatz
 \be
 \Theta=\Omega^\dagger\,\Omega
 \label{consy}
 \ee
where the
$N$ by $N$ matrix factor $\Omega$
(which is, naturally, non-unitary when $\Theta \neq I$)
is defined as composed of an $N-$plet of certain ``ketket''
column vectors
 \be
 {\Omega}^\dagger_{}=
 \left[ \begin {array}{cccc}
 |\psi_1\kkt\,,&|\psi_2\kkt\,,&\ldots&|\psi_N\kkt\,
 \end {array} \right]\,.
 \label{cuieNu}
 \ee
Interested readers may find the detailed explanation of the method in
\cite{PLA}. Here it is sufficient to point out only that
the column vectors $|\psi_j\kkt$
have to be eigenvectors of the Hermitian-conjugate
version of our non-Hermitian Hamiltonian,
 \be
 H^\dagger\,|\psi_n\kkt = E_n\,|\psi_n\kkt\,,
 \ \ \ \ n=1,2,\ldots\,,N\,.
 \label{qdirSE}
 \ee
One of the key merits of this approach is the
positivity of metric which
trivially follows from formula (\ref{consy}).
A disadvantage of this approach
lies in an excessive
length of formulae
describing our model
at large $A \approx \sqrt{2}$
(i.e., in the vicinity of EP3).
Even at $N=3$, therefore, it is difficult to display
the $A-$dependence of
the metric in print.

A help has been found in a
reparametrization
of dynamics with
$A=A(t)=\sqrt{2-2t^2}$.
Even though the display of metric
then still
requires the capacity exceeding a single printed page,
the computer-assisted
symbolic manipulations proved able to
offer the compact printable
formulae representing the eigenvalues $\theta_n(t)$
of the metric as functions of $t$
\textcolor{black}{up to
the unphysical EP2 limit of $t=0$. In this manner we can also speak about
an increase of non-Hermiticity with the decrease of
the absolute value $|t|$ of our} time-mimicking parameter.

\textcolor{black}{Let} us
remind the readers that
\textcolor{black}{when} we
fix the
normalization
of our column vector
solutions of Eq.~(\ref{qdirSE}), say, via their
last components,
 \be
 \br N|\psi_n\kkt = 1
 \label{houfnice}
 \ee
such
an {\it ad hoc\,} postulate
makes the metric unique.
As a typical sample of
the consequence of such a decision
(i.e., of such a form of removal of the ambiguity of the metric)
we obtained the following, easily proved
result.

\begin{lemma}
\label{ley}
After we choose $N=3$ and accept the normalization
convention (\ref{houfnice})
in Eq.~(\ref{cuieNu})
at all $n=1,2,\ldots,N$,
the triplet of eigenvalues of
the resulting metric (\ref{consy})
acquires the following explicit exact form,
 $$
  \theta_1=-3\,{t}^{2}+6-\sqrt
 {{t}^{4}-36\,{t}^{2}+36} \approx
 {\frac {2}{3}}{t}^{4}+{\frac {1}{3}}{t}^{6}+{\frac {11}{54}}{t}^{8}+\ldots \,,
\ \ \ \ \
  \theta_2=4\,{t}^{2}\,,
 $$
 $$
 \theta_3=-3\,{t}^{2}+6+\sqrt {{t}^{4}-36\,{t}^{2}+36}
 \approx
 12-6\,{t}^{2}-{\frac {2}{3}}{t}^{4}-{\frac {1}{3}}{t}^{6}-{\frac {11}
{54}}{t}^{8}+ \ldots
 \,.
 $$
\end{lemma}

It is worth adding that
at small $t\neq 0$,
\textcolor{black}{i.e., in the domain of
a maximal non-Hermiticity of the Hamiltonian,}
the latter eigenvalues
of our sample metric
are positive and
non-intersecting,
 $$
 0<
 \theta_1(t) \ll
 \theta_2(t) \ll
 \theta_3(t)  \,,\ \ \ \ \ 0<t^2 \ll 1\,.
 $$
This means that the metric
of Lemma \ref{ley}
(made unique by constraint (\ref{houfnice})
and denoted by a subscripted $S$
standing for ``sample'',  $\Theta^{(3)}_S(t)$)
satisfies all of the formal requirements of the theory.
In other words, the decrease of $|t|$
marks a \textcolor{black}{non-empty
corridor to extreme non-Hermiticity}
{\it alias\,} a
trajectory
of a smooth access to the EP3 singularity
along which the quantum system remains unitary.
\textcolor{black}{This also makes our construction
and interpretation of
the $t \gtrapprox 0$ family of our
maximally non-Hermitian
potentials $V^{(3)}(\sqrt{2-2t^2})$ completed.}

In the EP3 limit
the metric acquires a compact \textcolor{black}{but degenerate, singular}
form
 $$
 \lim_{t \to 0^+}\,\Theta^{(3)}_S(t)= \left[
 \begin {array}{ccc}
 3&-3\,i\sqrt {2}&-3
 \\\noalign{\medskip}3\,i\sqrt {2}&6&-3\,i\sqrt {2}
 \\\noalign{\medskip}-3&3\,i\sqrt {2}&3
\end {array} \right]\,
 $$
with eigenvalues 0, 0 and 12. Matrix $\Theta^{(3)}_S(0)$
ceases to be invertible and acceptable as a metric.
The quantum system in question \textcolor{black}{would cease} to be observable
\textcolor{black}{at $t=0$. Thus,}
one can speak about a
${\cal PT}-$symmetric quantum phase transition
\textcolor{black}{in such a limit}
\cite{CGbook,passage}.

\subsection{Spurious Hermiticity}

By construction, the
validity of
quasi-Hermiticity of our $N=3$ Hamiltonian
with respect to the above-selected sample metric
$\Theta^{(3)}_S(t)$
is not restricted to the small
values of \textcolor{black}{ $t^2\neq 0$}.
The eigenvalues
of Lemma \ref{ley}
remain positive,
ordered and distinct
(i.e., non-intersecting)
until
the onset of the
conventional Hermiticity of our
Hamiltonian matrix at
\textcolor{black}{the ``large''}
$t^2=1$.
Nevertheless,
our ``sample'' metric remains positive definite
even out of
the interval of ``times'' $t \in (-1,1)$.
As a metric, it
remains fully acceptable
in \textcolor{black}{an enhanced} interval of
 \be
 t \in (-t_{\max},t_{\max})\,,\ \ \ \ \
 t_{\max} \approx 1.014611872\,.
 \label{phyran}
 \ee
An explanation of apparent paradox is provided by
Figure \ref{globe}
in which we see the accidental coincidence
$\theta_2(1) =
 \theta_3(1)=4 $ of two eigenvalues
followed by their reordering,
 $$
 0<
 \theta_1(t) <
 \theta_3(t) <
 \theta_2(t)\,\ \ \ \ \ 1<|t|<t_{\max}\,.
 $$
This is an unexpected result.
The physical range (\ref{phyran}) of $t$
acquires a non-vanishing overlap
with the two
intervals of $t \in (-\infty,-1)$
and $t \in (1,\infty)$
in which our $\Theta_S^{(3)}-$quasi-Hermitian
Hamiltonian matrix is also Hermitian
in the sense of the linear algebra of textbooks,
i.e., in our present terminology,
also $\Theta_C^{(3)}-$quasi-Hermitian
with respect to the conventional identity-matrix
$\Theta_C^{(3)}=I^{(3)}$.

\begin{figure}[h]                    
\begin{center}                         
\epsfig{file=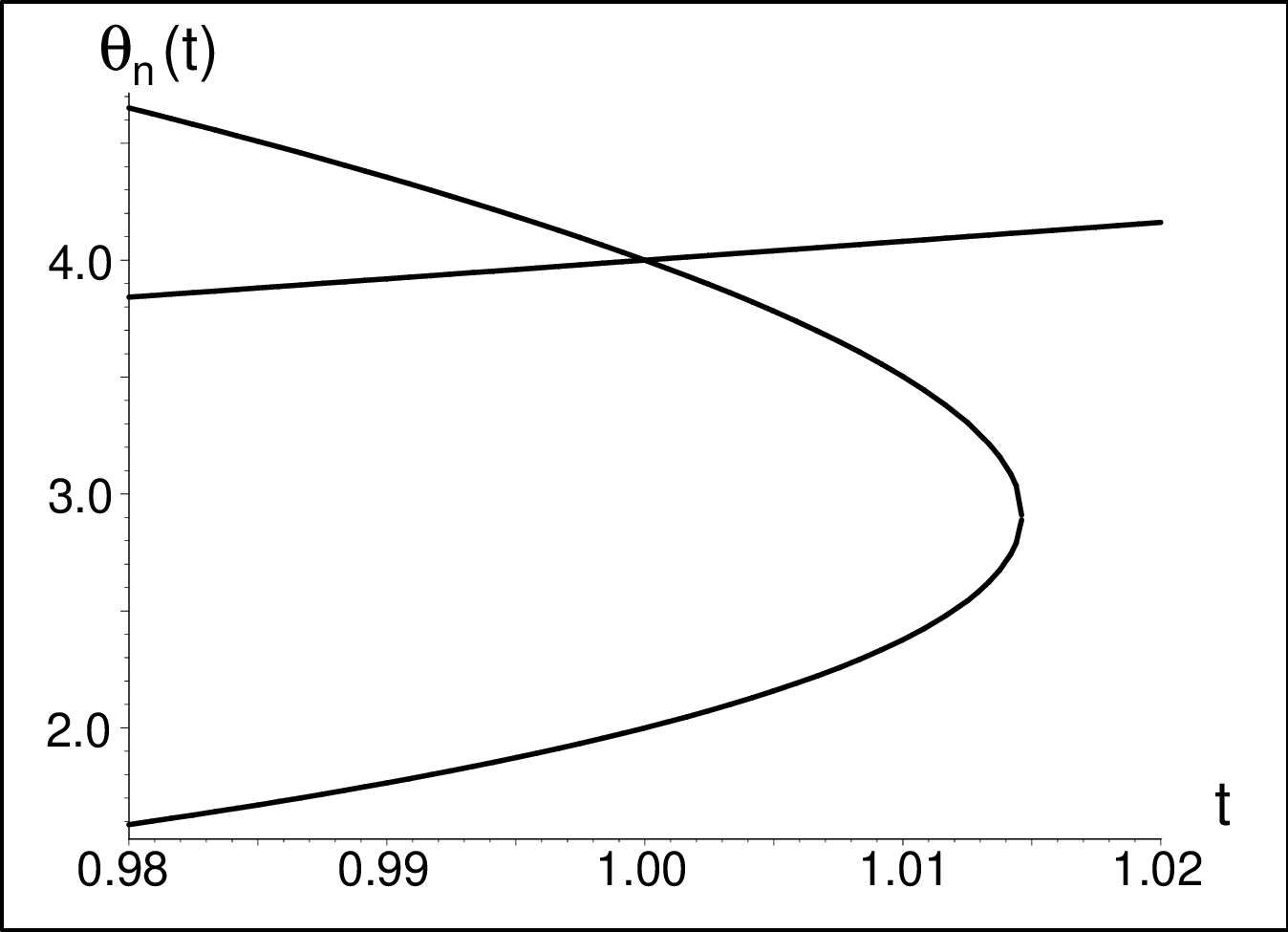,angle=270,width=0.35\textwidth}
\end{center}    
\caption{The triplet of eigenvalues of
our special metric $\Theta^{(3)}_S(t)$ in the vicinity of
the instant $t = 1$, beyond which the Hamiltonian matrix (\ref{han3}) itself becomes
Hermitian.
 \label{globe}}
\end{figure}

The origin of such an apparent duplicity
of alternative physical options and interpretations
is to be attributed to a
non-minimal anisotropy of the
physical Hilbert space geometry determined
by the inner-product metric  $\Theta^{(3)}_S(t)$.
This means that
at the instant $t=1$ of onset of the conventional
Hermiticity
of $H^{(3)}$
with respect to the trivial an trivially isotropic
metric $\Theta_C^{(3)}$
our reconstructed metric  $\Theta^{(3)}_S(t)$
reads
  $$
 \lim_{t \to 1}\,\Theta^{(3)}_S(t)=
 \left[ \begin {array}{ccc} 3&0&1\\\noalign{\medskip}0&4&0
 \\\noalign{\medskip}1&0&3\end {array} \right]\,.
 $$
With eigenvalues 4, 4 and 2, it remains anisotropic
(for a graphical reconfirmation
of the random coincidence of the two of its eigenvalues
see, once more, Figure \ref{globe}).

\section{Corridors to extreme non-Hermiticity at $N>3$
\label{relocika}}

In quantum mechanics of conventional textbooks \cite{Messiah}
the physical metric is trivial, $\Theta_C=I$.
This would make the
description of a quantum phase transition
impossible
because
in such a setting, all of the eigenvalues
of all of the observables  (say, $\Lambda$)
have the well known tendency of avoiding the degeneracy.
Thus, when we let such
an operator vary with a parameter, $\Lambda=\Lambda(g)$, its
observability status remains robust.
This implies that the
degeneracy
(i.e., the confluence of at least some of
the eigenvalues and eigenvectors)
can only occur at the
exceptional-point value of the parameter \cite{Kato}.

One of the main theoretical assumptions accepted
in our present paper
is that we are only considering
the
observable $N$ by $N$ matrix Hamiltonians of a
manifestly non-Hermitian and, at the same time,
very specific and restrictive
form of Eq.~(\ref{hamacek})
containing a {\em strictly local\,} potential (\ref{aVto}).
In such a case,
the realization of a {\em complete},
phase-transition-mimicking
EPN-related confluence of {\em all\,} of
the eigenvalues and eigenvectors
need not exist at all.
At every dimension $N$ of the Hilbert space,
this existence has to be demonstrated constructively.

\subsection{$N=4$\label{sublocika}}

New technical challenges emerge at $N=4$.
With
 \be
 H^{(4)}=
 \triangle^{(4)}+V^{(4)}(A,B)
 =\left[ \begin {array}{cccc} -iA&-1&0&0
 \\{}-1&-iB&-1&0\\{}0&-1&iB&-1\\{}0&0&-1&iA
 \end {array} \right]
 \label{hacty}
 \ee
the formula for energies is longer but still elementary,
 $$
 E_{\pm,\pm}=\pm 1/2\,\sqrt {6\pm 2\,\sqrt {5-2\,{A}^{2}-6\,{B}^{2}
 +{A}^{4}-2\,{B}^{2}{A}^{2}+{B}^{4}-8\,BA}-2\,{A}^{2}-2\,{B}^{2}}\,.
 $$
The search for the two EP4-determining parameters
becomes perceivably less straightforward,
requiring a slightly deeper inspection of secular equation
 \be
 {{\it E}}^{4}+ b\, {{\it E}}^{2}+c=0\,,
 \ \ \ b=b(A,B)= -3+{A}^{2}+{B}^{2} \,,
 \ \ \ c=c(A,B)
 =(1+A\,B)^2-A^2\,.
 \label{eku}
 \ee
The EP4 singularity at $E=0$ may only occur when
the secular equation degenerates to its
trivial form $E^4=0$, i.e., when $b=c=0$,
i.e., when we satisfy the following
pair of constraints
\be
-3+{A}^{2}+{B}^{2}=0\,,\ \ \ \
1+2\,
BA-{A}^{2}+{B}^{2}{A}^{2}=0\,.
\label{rucha}
\ee
In a brute-force approach we may eliminate
$B= \sqrt{3-A^2}$
and
get
the single reduced secular equation in non-polynomial form,
 $$
 {{\it E}}^{4}+1+2\,\sqrt {3-{A}^{2}}A+2\,{A}^{2}-{A}^{4}=0\,.
 $$
It is easily tractable numerically.
Out of the two
available
real and positive roots
let us first
consider the larger one,
 $$
 A^{(EP4)}= 1.683771565\,.
 $$
It appears to correspond to the smaller value of
 $$
  B^{(EP4)}= 0.4060952085\,,
 $$
so that the shape of our discrete potential
(and, first of all, the asymptotic growth
of its absolute value)
mimics the shape of its continuous ICO
analogue.

For our
above-declared purposes
of the maximalization of non-Hermiticity,
such a choice of the roots may be declared
preferred and, in this sense, unique.
The insertion of the parameters in the Hamiltonian
yields the canonical formula (\ref{relaJ})
in which the
transition matrix has been found to read
 $$
 Q^{(4)}= \left[ \begin {array}{cccc}   i&-
 1.835086683& - 1.683771565\,i& 1
 \\{} 1.683771562&
 2.089866772\,i&- 1& 0\\{}
 - 1.683771565\,i& 1& 0& 0
 \\{}- 1& 0& 0& 0\end {array} \right]\,.
  $$
In spite of the implicit presence
of numerical round-off errors
the insertion
of the latter formula
in relation~(\ref{relaJ}) confirms
its validity and demonstrates
that even the
calculations performed in the
finite-precision floating-point arithmetics
suffice for the EPN-localization purposes.

\subsection{Domain of parameters of
unbroken ${\cal PT}-$symmetry}

A very explicit outline
of a systematic, albeit rather tedious construction of the
``admissible'',
unitarity-preserving processes of the collapse or
unfolding of the
EPN degeneracy
may be found described, for any $N$,
in
our recent dedicated paper \cite{corridors}.
Such an approach to the proof
is universal
and enables one to construct a perturbation of $H^{(EPN)}$
due to which the
EPN singularity becomes connected
with some standard and fully regular family of Hamiltonians
by a
``unitary-access''
corridor.

At $N=4$, the study of secular equation (\ref{eku})
can help us
to ask whether there are perturbations
belonging to the rather narrow family of
Hamiltonians of the form of Eq.~(\ref{hacty}).
Thus, inside
a two-dimensional domain ${\cal D}$ of
parameters $A$ and  $B$
we may try to
search for a corridor
of unitary access to EP4
using potentials
which still remain local:
\textcolor{black}{It is precisely the locality requirement which makes the
problem nontrivial.}

%
%
\begin{figure}[h]     
\begin{center}                         
\epsfig{file=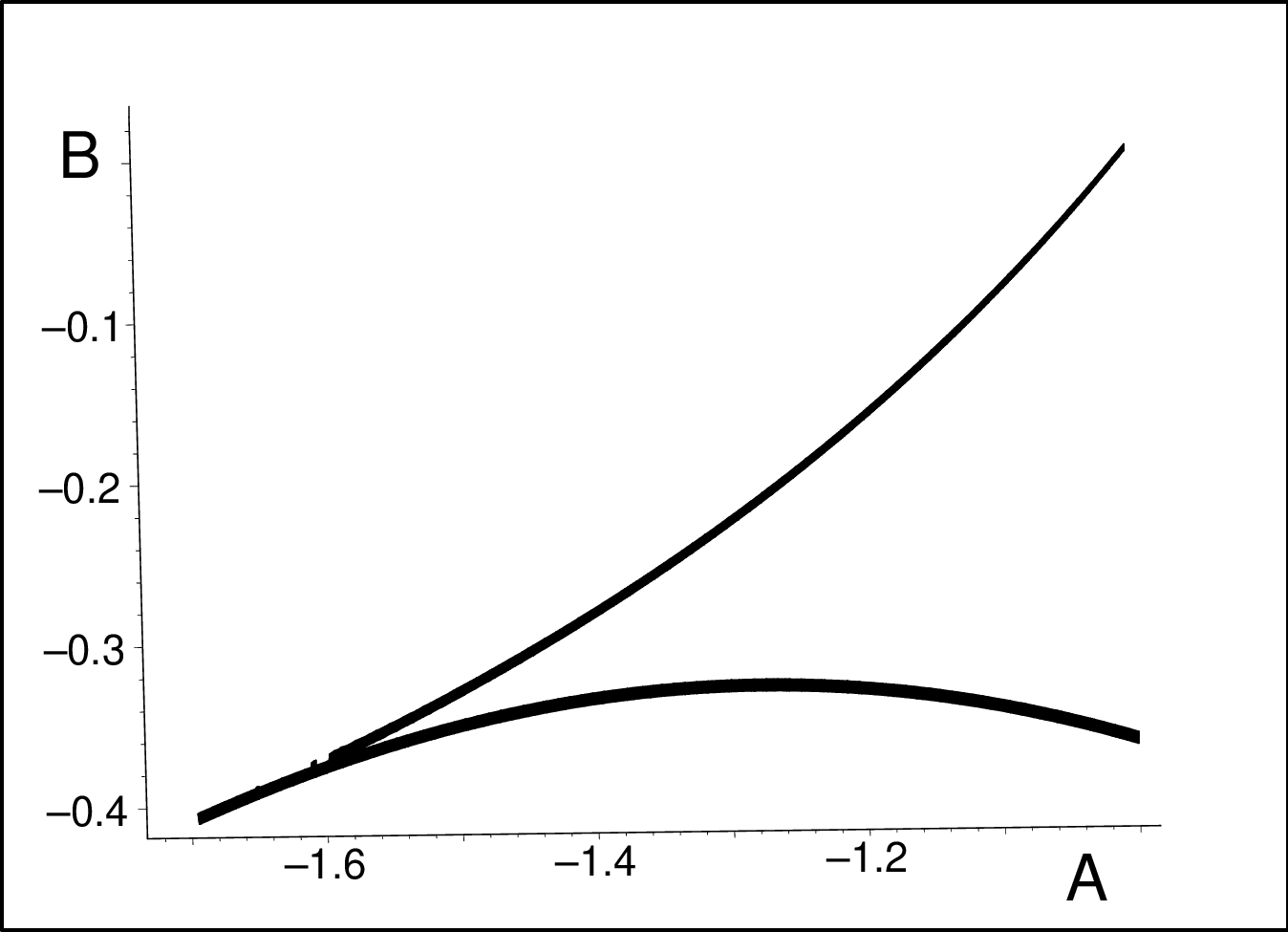,angle=270,width=0.35\textwidth}
\end{center}    
\caption{Spiked
form of the boundary of the physical
star-shaped domain ${\cal D}$ of
unbroken ${\cal PT}$ symmetry in the vicinity of
one of its four maximal-non-Hermiticity
extremes ($N=4$, numerical as well as analytic construction).\label{lobe}}
\end{figure}

\begin{lemma}
\label{lemma1}
For Hamiltonian (\ref{hacty})
the pairs of
values $A^{(EP4)}$ and $B^{(EP4)}$
exist and
lie on the boundary of a non-empty
physical domain ${\cal D}$
defined by the requirement that the related
spectrum of energies is
real, discrete and non-degenerate.
\end{lemma}

\begin{proof}
. In the vicinity of any EP4  the
acceptable
deviations of
parameters $A$ and $B$
from their critical EP4 values
must be
compatible with the requirement
that the root
$E^2$ of secular Eq.~(\ref{eku})
has to remain real, positive and non-degenerate.
This leads to the
triplet of conditions
 \be
 b(A,B)<0\,,
 \ \ \ c(A,B)>0\,,\ \ \ b^2(A,B)>4\,c(A,B)\,.
 \ee
Numerically, these inequalities may be shown to
determine
the ``narrowing corridor'' boundary of
the physical parametric domain
${\cal D}$ of the spiked
shape as sampled in Figure \ref{lobe}.

A non-numerical version of the proof of existence of
such a corridor of access to the
maximally non-Hermitian quantum-dynamical regime
is also feasible because
near the EP4 extreme
the interior of ${\cal D}$
can be reparametrized
in terms of three new real,
positive and not too large
auxiliary parameters $\alpha$, $\beta$ and $\gamma$,
 \be
 A^2+B^2=3\,\cos^2 \alpha\,,\ \ \
 1+A\,B=A\,\cosh \beta\,,\ \ \
 3\,\sin^2 \alpha = 2\,A\,\sinh \beta\,\cos \gamma\,.
 \ee
Once we eliminate $B=\cosh \beta -1/A$ and $3\,\sin^2 \alpha$
we arrive at a single constraint
 \be
 P(A,\beta,\gamma)
 =A^4+(A\,\cosh \beta-1)^2-3+2\,A\,\sinh \beta \cos \gamma=0\,.
 \label{constr}
 \ee
In the limit $\beta \to 0$ this relation
degenerates to the
(implicit) definition of the correct
EP4 root.
In its vicinity,
once we choose $A^{(EP4)}= 1.683771565$,
the curve $P(A,0,0)$ can be shown to
grow as a function of $A$.
The inclusion of the small
corrections
${\cal O}(\beta^2)$ and
${\cal O}(\beta\,\gamma^2)$ only shifts the curve
slightly upwards.
This forces the root of Eq.~(\ref{constr})
to move slightly to the left,
i.e.,
the
value of $A$ (i.e., of
the matrix element
of the regularized potential) only
gets slightly smaller and stays real.
In parallel,
also the related value of $B$ remains real.
These observations can be rephrased as the proof of
existence of the $(\beta,\gamma)-$parametrized
unitarity-preserving corridor in which the
EP4 degeneracy can unfold.
\end{proof}

%
\begin{figure}[h]                    
\begin{center}                         
\epsfig{file=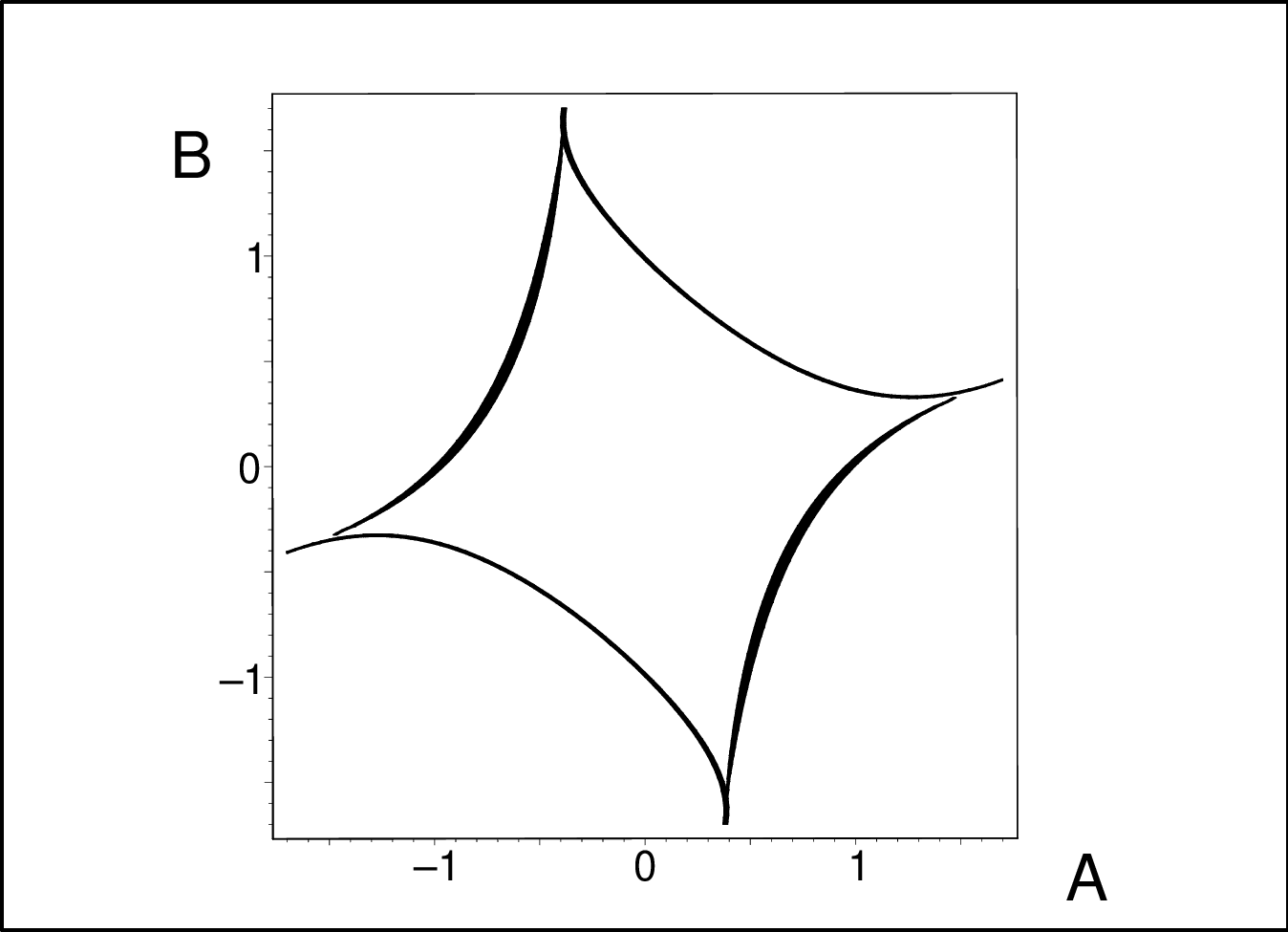,angle=270,width=0.35\textwidth}
\end{center}    
\caption{Two-dimensional physical open
domain ${\cal D}$ and its four maximal-non-Hermiticity
extremes ($N=4$, numerical construction).
 \label{urlobe}}
\end{figure}

We may conclude that near the EP4 singularity, i.e.,
in the dynamical regime
characterized by the maximal non-Hermiticity
of the potential
the corridor of admissible unitary unfoldings of the singularity
is determined by
a two-parametric perturbation
of the singular EP4 limit.
Near one of the four EP4 singularities
the sample of such a shape is displayed in
Figure \ref{lobe} where the inequalities
  $$
  A \gtrapprox A^{(EP4)}_-\approx  -1.683771565\,,\ \ \ \ \
B \gtrapprox  B^{(EP4)}_-\approx - 0.4060952085\,
 $$
define the
physical
corridor
locally.
One might only add and emphasize
that the picture clearly demonstrates
that such a corridor becomes truly narrow
near its maximal non-Hermiticity EP4 extreme.
Naturally, the same feature of ${\cal D}$
can be expected to be observed at any $N$.

Inside the corridor the quantum system is endowed with
the conventional probabilistic and unitary-evolution
interpretation.
What remains to be analyzed and described
is the form of the rest of the physical domain ${\cal D}$
which lies far from the singularity.
The task
remains elementary
and is left to the readers as an exercise.
Its result is displayed in
Figure \ref{urlobe}
showing that at $N=4$
the physical domain has a shape
of a deformed square
with protruded, spike-shaped vertices.

\subsection{$N=5$
\label{rerelocika}}

With the growth of $N$ the implementation of the
construction strategy based
on the brute-force elimination of the first parameter
${A}={A}(B,C,\ldots)$ from the second item
in the hierarchy of equations (i.e., from
Eq.~(\ref{rucha}) at $N=4$,
etc) appears naive. A more advanced
approach (viz., the application of systematic
Groebner-based elimination technique)
has to be used at the larger $N>4$.
Even at $N=4$ this would  \textcolor{black}{simplify}
the construction
and yield, naturally, the same
numerical value of $B^{(EP4)}$ as a (unique)
positive and real
root of polynomial $B^4-2\,B^3-2\,B^2+6\,B-2$.

Once we move to the next model with $N=5$, the choice of
Hamiltonian
$$
 H^{(5)}=\triangle^{(5)}+V^{(5)}(A,B)
 =\left[ \begin {array}{ccccc} -iA&-1&0&0&0
 \\{}-1&-iB&-1&0&0\\{}0&-1&0&-1&0\\{}0&0&-1&iB&-1
\\{}0&0&0&-1&iA\end {array} \right]
$$
leads to the secular equation
$${{\it E}}^{5}- \left( 4-{A}^{2}-{B}^{2} \right) {{\it E}}^{3}-
 \left( -3-2\,BA+2\,{A}^{2}-{B}^{2}{A}^{2} \right) {\it E}=0
$$
yielding the constant vanishing $E_0=0$ and
the remaining quadruplet of nontrivial levels
 $$
 E_{\pm,\pm}=\pm
 1/2\,\sqrt {8 \pm 2\,\sqrt {4-8\,{B}^{2}+{A}^{4}-2\,{B}^{2}{A}^{2}+{B}^{4}
 -8\,BA}-2\,{A}^{2}-2\,{B}^{2}}\,.
 $$
Once we fix $B^{}= \sqrt{4-A^2}$
we get a reduced secular equation
$$
{{\it E}}^{}= -3-2\,\sqrt {4-{A}^{2}}A-2\,{A}^{2}+{A}^{4}
$$
etc. Again, a more efficient
Groebner-elimination technique leads to
the identification of
$B^{(EP5)}$ with the root of polynomial
 $$
 B^4-2\,B^3-4\,B^2+6\,B+5\,.
 $$
The result is $B^{(EP5)}= 0.6683178062$
yielding
$A^{(EP)}=1.885033504 $. The EP5 degeneracy
becomes
confirmed by the verification of relation
 $$
 Q^{(5)}J^{(5)}=  H^{(EP5)}\,Q^{(5)}
 $$
i.e.,
by the criterion
as provided by
Eq.~(\ref{relaJ}) above.

The rest of the $N=5$ construction remains analogous
to its $N=4$ predecessor.

\section{The emergence of ambiguities at $N\geq 6$
\label{rererelocika}}

The first serious symbolic-manipulation challenges
emerge at $N=6$.
The Hamiltonian
varies with the triplet of parameters $A,B$ and $C$
so that the
secular polynomial is of the third order in $E^2$.
The energies are obtainable in closed form. Nevertheless,
the formulae (a.k.a. Cardano formulae)
are far from compact so that
for the majority of purposes
we recommend
to use the brute-force
numerical localization of the bound-state energy roots.

The parallel
task of the localization
of the three real and positive EPN-generating coupling constants
$A^{(EP6)},B^{(EP6)}$ and $C^{(EP6)}$ in $H^{(EP6)}$
remains tractable non-numerically because it may
be reduced to the coupled triplet of polynomial algebraic equations
 $$
-5+{A}^{2}+{B}^{2}+{C}^{2}=0\,,\ \ \ \
 6-{B}^{2}+2\,BA-2\,{C}^{2}-3\,{A}^{2}+{B}^{2}{A}^{2}+{C}^{2}{A}^{2}
 +2\,CB+{C}^{2}{B}^{2} =0\,,\ \ \
 $$
 $$
 -1+{C}^{2}-2\,BA+2
 \,CA+2\,{C}^{2}BA+{A}^{2}+2\,CB{A}^{2}-{B}^{2}{A}^{2}+{B}^{2}{C}^{2}{A}^{2}
 =0\,.
 $$
The brute-force numerical
search for the roots of this system
appeared inefficient.
In contrast, the application of the
computer-assisted algebraic
Groebner elimination technique proved
straightforward. It
led to the result
with an amazing computing-time efficiency.

Two features of
the construction should be mentioned.
First,
the determination of the
real and positive value of $B^{(EP6)}$
appeared reduced to the localization of a root
of a fairly complicated Groebnerian polynomial
$ P^{(N)}(B)$
having an unexpectedly high degree,
 $$
 P^{(6)}(B)=
 B^{23}-2\,B^{22}-20\,B^{21}+32\,B^{20}+
 188\,B^{19}-216\,B^{18}-1060\,B^{17}
 +768\,B^{16}+
 $$
 $$
 +3782\,B^{15}-1308\,B^{14}-8492\,B^{13}+16\,B^{12}+11164\,B^{11}
 +4008\,B^{10}-6668\,B^{9}-7072\,B^{8}-703\,B^{7}+
 $$
 $$
 +5678\,B^6
 +2320\,B^5-1200\,B^4-1248\,B^3-96\,B^2+160\,B+32\,.
 $$
Another surprise was that
the results ceased to be unique.
Polynomial $P^{(6)}(B)$ appeared to provide
two alternative real values  for the eligible
EP6-supporting parameter, viz,
$B^{(EP6)}_{expected}= 0.8635733388$ and
$B^{(EP6)}_{unexpected}=0.4333101655$.
In a way guided by our initial model-building
concept of discretization of the continuous
ICO potential we decided to prefer
the $N-$dependence of the critical parameters
of our immediate interest which would be monotonous
rather than oscillatory.

\begin{table}[h]
\caption{Real and positive EPN-supporting parameters
at the first few $N$.}
\vspace{0.21cm}
 \label{Pwe2}
\centering
\begin{tabular}{||c||c|c|c|c|c||}
 \hline \hline
 $N$
 &$2$
 &$3$
 &$4$
 &$5$
 &$6$
 \\
 \hline \hline
 $A^{(EP6)}$
 &$1.000$
 &$1.414$
 &$1.684$
 &$1.885$
 &$2.046$
 \\
 \hline \hline
 $B^{(EP6)}$
 &-
 &-
 &$0.406$
 &$0.608$
 &$0.864$
  \\
 \hline \hline
 $C^{(EP6)}$
 &-
 &-
 &-
 &-
 &$0.261$
  \\
 \hline \hline
\end{tabular}
\end{table}

Having this requirement in mind we found it confirmed by
the first four columns of Table \ref{Pwe2}.
On this ground
we felt entitled to restrict our choice to
the larger root $B^{(EP6)}_{expected}$.
After the routine backward insertions we
obtained our ultimate numerical result
 $$
 A^{(EP6)}= 2.046061191\,,\ \ \
 B^{(EP6)}=0.8635733388\,,\ \ \
 C^{(EP6)}=0.2605285271\,.
 $$
It
fits very well our {\it ad hoc\,} monotonicity hypothesis.
Our requirement of
monotonicity
seems to be well founded, therefore, reflecting
a tacit reference to
the smooth ICO benchmark.

In this light, various extensions of our analysis
to the models with $N>6$
may be expected straightforward.
Marginally, let us add that
our hypothesis of monotonicity
could have also been introduced
just for the sake of brevity of our message.
Alternatively, our results
(and, in particular, the list of Table \ref{Pwe2})
could be, in principle, extended
to cover also the
oscillatory shapes of the
EPN-admitting
discrete imaginary potentials.

Even under the monotonicity hypothesis
there might exist
multiple families of
discrete imaginary potentials
simulating the IEP singularity via
its EPN analogue
at large $N$.
Such
a project seems
difficult but supported by
equation Nr.~17 of paper \cite{Siegl}:
The authors
proposed there the
existence of the
IEP anomalies
in a much broader class of
continuous
${\cal PT}-$symmetric alternatives to ICO.

\section{Conclusions\label{presumma}}

The
problem addressed in our present paper
concerned the existence and
the nature of correspondence between the
reality of the bound state energies and the
degree of non-Hermiticity of the potential,
i.e., of its deviation from the conventional
real one.
The answers have been sought via
the localization, shape and extremes
of boundaries of the
physical
domain ${\cal D}$ of the
unitarity-compatible  parameters $ A, B, \ldots$
specifying the potential.

We \textcolor{black}{believe} that
our constructive study of various
discrete versions
of Schr\"{o}dinger equation
\textcolor{black}{throws} new light upon the
absence of the Riesz basis features of the bound-state-like
wave functions
generated by the enigmatic but extremely elementary and
popular continuous ICO benchmark model (\ref{lico}).
We can speak about success because
our multiparametric
discrete imaginary-potential
quantum bound state model
(\ref{hamacek}) -- (\ref{aVto})
appeared sufficiently flexible, especially as a tool of
an illustration
of several typical
though still fairly counterintuitive
mathematical as well as phenomenological features of
non-Hermitian Hamiltonians with real spectra.

The main mathematical result of our study
should be seen in the demonstration,
for the strictly local (albeit discrete)
and purely imaginary potentials,
of the existence
and of a unitary-evolution accessibility
of the exceptional points of
the maximal order $N$.
Indeed,
once we accepted the
unbroken ${\cal PT}-$symmetry constraint
we found that
at the first few nontrivial Hilbert-space dimensions $N$
at least,
the real and
unitarity-compatible ``physical''
parameters $A, B, \ldots$ of the potentials
lie
inside an open domain ${\cal D}$,
the shape of
the boundaries of
which
admits a comparatively feasible description
in a maximally \textcolor{black}{non-Hermitian} dynamical regime,
i.e., near one of
the EPN singularities.

In a purely formal sense it has been found satisfactory
to see that the construction of the corresponding
``unphysical limit'' parameters $A^{(EPN)}, B^{(EPN)}, \ldots$
themselves
appeared either non-numerical (up to $N=5$) or,
with the ample assistance of
computer-assisted symbolic and numerical manipulations,
straightforward.
This being said,
another, slightly less gratifying
observation was that
the later values $A^{(EPN)}, B^{(EPN)}, \ldots$
had to be deduced from the roots of a
polynomial, the degree of
which appeared to grow very quickly with
the growth of the dimension $N$
of the Hilbert space in question.
This means that any extrapolation of our present
toy-model constructions
(with the discrete imaginary
potentials of maximal non-Hermiticity as
sampled in Table \ref{Pwe2})
to any hypothetical continuous-coordinate limit $N \to \infty$
does not seem to be feasible at present.


\newpage

\begin{thebibliography}{99}


\bibitem{BB}
C. M. Bender and S. Boettcher,
``Real Spectra in
Non-Hermitian Hamiltonians having PT Symmetry,''
Phys. Rev. Lett. {\bf 80},
5243--5246
(1998).


\bibitem{Carl}
C. M. Bender,
``Making Sense of Non-Hermitian Hamiltonians,''
Rep. Prog. Phys.  {\bf 70}, 947--1018 (2007).

\bibitem{DB}
D. Bessis, private communication (1992).

\bibitem{Siegl}
P. Siegl and D. Krej\v{c}i\v{r}\'{\i}k,
%
%
``On the metric operator for the imaginary cubic oscillator,''
Phys. Rev. D {\bf 86},  121702(R) (2012).
%



\bibitem{Uwe}
U. G\"{u}nther and F. Stefani,
``IR-truncated PT -symmetric $ix^3$ model and its asymptotic
spectral scaling graph,'' arXiv 1901.08526
(2019).



\bibitem{Dieudonne}
J. Dieudonn\'{e},
``Quasi-Hermitian Operators,'' in {Proc. Int. Symp.
Lin. Spaces}, (Pergamon: Oxford, UK, 1961), {pp. 115--122}.
%


\bibitem{Geyer}
F. G. Scholtz, H. B. Geyer, and F. J. W. Hahne,
``Quasi-Hermitian
Operators in Quantum Mechanics and the Variational Principle,''
Ann. Phys. (NY) {\bf 213},
74--101 (1992).
%



\bibitem{foundations}
M. Znojil,
``\textcolor{black}{The} intrinsic exceptional point -- a challenge in quantum theory,''
\textcolor{black}{Foundations {\bf 5},  8 (2025)}
(arXiv:2411.12501).




\bibitem{ali}
A. Mostafazadeh,
``Pseudo-Hermitian Representation of Quantum
Mechanics,''  
Int. J. Geom. Meth. Mod. Phys. {\bf 7},
1191--1306 (2010).



\bibitem{book}
%
\emph{Non-Selfadjoint Operators in Quantum Physics: Mathematical
Aspects}, edited by
F. Bagarello, J.-P. Gazeau, F. H. Szafraniec, and M. Znojil
(Wiley, Hoboken, NJ, 2015).



\bibitem{SIGMA}
M. Znojil,
``Three-Hilbert-space formulation of Quantum Mechanics,''
{Symm. Integ. Geom. Meth. Appl. (SIGMA)}   {\bf  5}, 001
(2009);
%
(arXiv: 0901.0700).


\bibitem{Kato}
T. Kato, {\em Perturbation Theory for Linear
Operators}
(Spinger, Berlin,
1966).



\bibitem{catast}
M. Znojil,
``Quantum catastrophes: a case study,''
{J. Phys. A: Math. Theor.}   {\bf 45},
444036  (2012).


\bibitem{SIGMAdva}
M. Znojil,
``On the role of the normalization factors $\kappa_n$ and
of the pseudo-metric P in crypto-Hermitian quantum models,''
{Symm. Integ. Geom. Meth. Appl. (SIGMA)}   {\bf  4}, 001
(2008);
(arXiv: 0710.4432v3).
%
%


\bibitem{Lotor}
D. Krej\v{c}i\v{r}\'{\i}k,
V. Lotoreichik, and M. Znojil,
``The minimally anisotropic metric operator in quasi-hermitian quantum
mechanics,''
{Proc. Roy. Soc. A: Math. Phys. Eng. Sci.}
  {\bf  474},
      20180264 (2018).



\bibitem{PLA}
M. Znojil,
``Non-Hermitian-Hamiltonian-induced unitarity and optional
physical inner products in Hilbert space,''
Phys. Lett. A  {\bf 523}, 129782 (2024).


\bibitem{CGbook}
E. Caliceti and S. Graffi,
``Criteria for the reality of the spectrum
of PT-symmeric Schr\"{o}dinger operators and for the existence
of PT-symmetric phase,''
in
F. Bagarello, J.-P. Gazeau, F. H. Szafraniec, and M. Znojil, Eds,
\emph{Non-Selfadjoint Operators in Quantum Physics: Mathematical
Aspects}
(Wiley, Hoboken, NJ, 2015).

\bibitem{passage}
M. Znojil,
``Passage through exceptional point: case study,''
%
%
Proc. Roy. Soc. A:
Math., Phys. \& Eng. Sci. {\bf  476},
  20190831
(2020).


\bibitem{Messiah}
A. Messiah,  \emph{Quantum Mechanics}
(North Holland, Amsterdam, The Netherlands, 1961).

\bibitem{corridors}
M. Znojil,
``Unitarity corridors to exceptional points,''
%
Phys. Rev. A {\bf  100},
032124
(2019).



\end{thebibliography}
\end{document}